\documentclass{article}
\usepackage{graphicx} % Required for inserting images
\usepackage[dvipsnames]{xcolor}
\usepackage{amsmath}
\usepackage{booktabs}
\usepackage{colortbl}
\usepackage{threeparttable}
\usepackage{tabularx}
\usepackage[referable]{threeparttablex} % footnotes in tabu
\newcolumntype{R}{>{\raggedleft\arraybackslash}X}
\newcolumntype{L}{>{\raggedright\arraybackslash}X}
\newcolumntype{C}{>{\centering\arraybackslash}X}
\newcolumntype{A}{>{\columncolor{gray!25}}C}
\newcolumntype{a}{>{\columncolor{gray!25}}c}

\usepackage{natbib}
\RequirePackage[T1]{fontenc} \RequirePackage[tt=false, type1=true]{libertine} 

\usepackage[a4paper, total={6.5in, 9.5in}]{geometry}
\usepackage{amsmath,amsthm,amssymb}
\usepackage[dvipsnames]{xcolor} % for fancy colors
\usepackage{setspace} % To set line spacing

\usepackage{hyperref}
\hypersetup{
    colorlinks=true,
    citecolor=Blue,
    linkcolor=Blue,
    filecolor=Plum,      
    urlcolor=CadetBlue,
    pdftitle={Overleaf Example},
}

\theoremstyle{plain}
\newtheorem{proposition}{Proposition}

\def\keywords{\vspace{.5em} % Add keywords
{\noindent \textit{Keywords}: }}

\def\JEL{\vspace{.5em} % Add keywords
{\noindent \textbf{\emph{JEL} classification number}: }}

\def\AMS{\vspace{.5em} % Add keywords
{\noindent \textbf{\emph{MSC} class}: }}

\title{Ranking matters: Does the new format select the best teams for the knockout phase in the UEFA Champions League?}
\author{
L\'aszló Csat\'o\thanks{~Institute for Computer Science and Control (SZTAKI), Hungarian Research Network (HUN-REN), Laboratory on Engineering and Management Intelligence, Research Group of Operations Research and Decision Systems, Budapest, Hungary \newline
Corvinus University of Budapest (BCE), Institute of Operations and Decision Sciences, Department of Operations Research and Actuarial Sciences, Hungary \newline
\texttt{laszlo.csato@uni-corvinus.hu}} \and 
Karel Devriesere\thanks{~Ghent University, Department of Business Informatics and Operations Management, Belgium \newline
FlandersMake@UGent -- core lab CVAMO, Ghent, Belgium \newline
\texttt{karel.devriesere@ugent.be}} \and
Dries Goossens\thanks{~Ghent University, Department of Business Informatics and Operations Management, Belgium \newline
FlandersMake@UGent -- core lab CVAMO, Ghent, Belgium \newline
\texttt{dries.goossens@ugent.be}} \and
Andr\'as Gyimesi\thanks{~University of P\'ecs, Faculty of Business and Economics, Hungary \newline
Institute for Computer Science and Control (SZTAKI), Hungarian Research Network (HUN-REN), Laboratory on Engineering and Management Intelligence, Research Group of Operations Research and Decision Systems, Budapest, Hungary \newline
\texttt{gyimesi.andras@ktk.pte.hu}} \and 
Roel Lambers\thanks{~Eindhoven University of Technology, Department of Mathematics and Computer Science, The Netherlands \newline
\texttt{r.lambers@tue.nl}} \and
Frits Spieksma\thanks{~Eindhoven University of Technology, Department of Mathematics and Computer Science, The Netherlands \newline
\texttt{f.c.r.spieksma@tue.nl}}
}
\date{\today}

\begin{document}

\maketitle

\begin{abstract}
\noindent
Starting in the 2024/25 season, the Union of European Football Associations (UEFA) has fundamentally changed the format of its club competitions: the group stage has been replaced by a league phase played by 36 teams in an incomplete round robin format.
This makes ranking the teams based on their results challenging because teams play against different sets of opponents, whose strengths vary. 
In this research note, we apply several well-known ranking methods for incomplete round robin tournaments to the 2024/25 UEFA Champions League league phase in order to check the robustness of the official ranking, as well as to call the attention of organizers to the non-trivial issue of ranking in these competitions. Our results show that it is doubtful whether the currently used point-based system provides the best ranking of the teams.

\keywords{fairness; incomplete round robin tournament; ranking; strength of schedule; UEFA Champions League}

\AMS{15A06, 91B14}
% Linear equations (linear algebraic aspects)
% Social choice

\JEL{C44, Z20}
% Operations Research, Statistical Decision Theory
% Sport Economics: General
\end{abstract}

\section{Introduction}

The 2024/25 UEFA Champions League has introduced a new format, replacing the traditional group stage with an incomplete round robin tournament where each of the 36 teams plays against eight different opponents. The official ranking is a lexicographical order based on the number of points, followed by goal difference, and the number of goals scored as tiebreakers. While using four strength-based pots is intended to create opponent sets of comparable strength, the difficulty of a team's set of opponents -- known as strength of schedule (SoS), see, e.g.\ \citet{Fearnhead2010} -- can still vary considerably between different teams.

\begin{table}[t!]
    \centering
    \caption{Strength of schedules in the 2024/25 UEFA Champions League league phase}
    \label{Table1}
    \begin{threeparttable}
	\rowcolors{1}{}{gray!20}	
       \begin{tabularx}{\textwidth}{l CCCC CCCC CR} \toprule
    Team & \multicolumn{2}{c}{Pot 1} & \multicolumn{2}{c}{Pot 2} & \multicolumn{2}{c}{Pot 3} & \multicolumn{2}{c}{Pot 4} & Sum & SoS \\ \bottomrule
    Liverpool  & 15  & 3  & 16  & 15  & 16  & 14  & 6  & 3  & 88  & 85  \\
    Barcelona  & 15  & 15  & 15  & 13  & 0  & 6  & 13  & 13  & 90  & 86  \\
    Arsenal  & 13  & 19  & 7  & 15  & 11  & 11  & 13  & 3  & 92  & 88  \\
    Inter Milan  & 3  & 11  & 19  & 16  & 6  & 0  & 13  & 4  & 72  & 68  \\
    Atleti  & 3  & 13  & 16  & 13  & 16  & 3  & 0  & 4  & 68  & 62  \\
    Leverkusen  & 19  & 21  & 15  & 18  & 3  & 13  & 4  & 13  & 106  & 99  \\
    Lille  & 15  & 21  & 12  & 18  & 13  & 11  & 6  & 6  & 102  & 95  \\
    Aston Villa  & 15  & 3  & 12  & 11  & 12  & 0  & 6  & 13  & 72  & 65  \\
    Atalanta  & 15  & 19  & 19  & 7  & 12  & 0  & 6  & 10  & 88  & 82  \\
    B.~Dortmund  & 19  & 15  & 7  & 11  & 12  & 11  & 6  & 6  & 87  & 78  \\
    Real Madrid  & 15  & 21  & 15  & 15  & 3  & 16  & 10  & 13  & 108  & 99  \\
    Bayern M\"unchen  & 13  & 19  & 13  & 7  & 11  & 13  & 0  & 16  & 92  & 83  \\
    Milan  & 21  & 15  & 11  & 16  & 6  & 11  & 3  & 0  & 83  & 74  \\
    PSV  & 21  & 13  & 7  & 12  & 11  & 6  & 3  & 13  & 86  & 78  \\
    Paris  & 11  & 15  & 18  & 19  & 14  & 3  & 3  & 10  & 93  & 83  \\
    Benfica  & 19  & 15  & 18  & 12  & 13  & 6  & 6  & 13  & 102  & 92  \\
    Monaco  & 19  & 19  & 13  & 19  & 6  & 11  & 16  & 6  & 109  & 99  \\
    Brest  & 15  & 19  & 16  & 7  & 14  & 3  & 6  & 4  & 84  & 74  \\
    Feyenoord  & 15  & 11  & 16  & 13  & 3  & 16  & 4  & 3  & 81  & 71  \\
    Juventus  & 11  & 3  & 13  & 11  & 14  & 16  & 10  & 16  & 94  & 85  \\
    Celtic  & 3  & 15  & 11  & 15  & 0  & 11  & 0  & 16  & 71  & 62  \\
    Man City  & 19  & 13  & 11  & 12  & 13  & 11  & 4  & 0  & 83  & 72  \\
    Sporting CP  & 11  & 3  & 19  & 11  & 16  & 14  & 6  & 6  & 86  & 75  \\
    Club Brugge  & 15  & 11  & 12  & 15  & 11  & 12  & 16  & 6  & 98  & 87  \\
    GNK Dinamo  & 15  & 15  & 15  & 19  & 12  & 3  & 13  & 0  & 92  & 81  \\
    Stuttgart  & 13  & 15  & 15  & 12  & 0  & 6  & 4  & 0  & 65  & 52  \\
    Shakhtar  & 15  & 15  & 15  & 19  & 0  & 14  & 13  & 6  & 97  & 81  \\
    Bologna  & 15  & 21  & 7  & 13  & 16  & 11  & 13  & 16  & 112  & 97  \\
    Crvena Zvezda  & 19  & 19  & 13  & 15  & 14  & 0  & 10  & 13  & 103  & 85  \\
    Sturm Graz  & 3  & 15  & 11  & 15  & 11  & 16  & 3  & 13  & 87  & 69  \\
    Sparta Praha  & 19  & 11  & 18  & 16  & 3  & 13  & 13  & 10  & 103  & 84  \\
    Leipzig  & 21  & 19  & 12  & 18  & 11  & 12  & 16  & 6  & 115  & 94  \\
    Girona  & 21  & 13  & 19  & 15  & 13  & 14  & 0  & 6  & 101  & 80  \\
    Salzburg  & 13  & 15  & 18  & 16  & 11  & 13  & 13  & 4  & 103  & 82  \\
    S.~Bratislava  & 11  & 15  & 15  & 18  & 11  & 12  & 10  & 3  & 95  & 71  \\
    Young Boys  & 19  & 19  & 15  & 7  & 6  & 12  & 16  & 10  & 104  & 80  \\ \bottomrule
    \end{tabularx}
    \begin{tablenotes} \footnotesize
    	%\item Probabilities are based on 1 million simulation runs.
        \item
        Each team has played against two different opponents from each pot, one at home and another away. The two columns Pot $i$ contain the number of points scored by the opponents from Pot $i$ in this order. For example, Liverpool has played against Real Madrid (that scored 15 points) at home and against Leipzig (that scored 3 points) away from Pot 1.
    	\item
        The column Sum gives, for each team, the total number of points collected by its eight opponents.
    	\item
        The column SoS shows, for each team, the difference between Sum and the number of points the team lost against its eight opponents.%is the sum of opponents' points minus the sum of points lost by the team against them.
    \end{tablenotes}
\end{threeparttable}
\end{table}

This is evident from Table~\ref{Table1}, which presents the number of points obtained by the opponents of each team (two from each pot). We calculate the strength of schedule of a team as the sum of its opponents' points, minus the points the team itself lost against them. Notably, \emph{Man City} and \emph{Sporting CP}, both of whom survived the league stage, faced significantly weaker opponents than \emph{GNK Dinamo}, which obtained the same number of points but was eliminated.
Furthermore, \emph{GNK Dinamo} was ranked below these teams based on goal difference, which was highly negative primarily due to their heavy defeat (2-9) against \emph{Bayern M\"unchen} in the first round. However, teams that face weaker opponents have an advantage regarding goal difference since securing high-margin victories is easier against weaker opponents, which again raises questions of fairness.
Analogously, in the race for the Round of 16, placing \emph{Aston Villa} in the Top 8 appears to have been strongly influenced by its relatively easy schedule. This stands in contrast to teams like \emph{Atalanta}, \emph{Borussia Dortmund}, \emph{Real Madrid}, \emph{Bayern M\"unchen}, and \emph{Milan}, all of them scoring just one point less.

The current note aims to explore alternative ranking methods that provide a different perspective on team performances by taking the strength of the schedule into account. By applying these approaches to the final league-phase standing in the 2024/25 UEFA Champions League, we illustrate how team rankings vary considerably depending on the chosen methodology. As the final ranking determines which teams qualify for the next phase, these results show the non-negligible effect of the ranking method on the fairness of the competition, and can call the attention of the decision makers to the issue of ranking in incomplete round robin tournaments. 

\section{Alternative ranking methods} \label{Sec2}

As an alternative to the official point-based ranking, we examine four classes of ranking methods suggested for incomplete round-robin tournaments, where teams face varying strengths of schedule: the direct ranking method (Section~\ref{Sec21}), Colley's ranking (Section~\ref{Sec22}), the generalized row sum (Section~\ref{Sec23}), and the least squares method (Section~\ref{Sec24}). Each of these ranking methods is derived by solving a system of linear equations that gives the strength $r_i$ of each team $i$. Note that other ranking methods exist, see, e.g.\ \citet{DevlinTreloar2018}, \citet{vaziri2018properties}, and \citet{DabadghaoVaziri2022}, as well as various metrics to express strength of schedule (see e.g.\ \citet{DeHollanderKarwan2025}).
Last but not least, Section~\ref{Sec25} discusses some connections between the ranking methods presented in Sections~\ref{Sec21}--\ref{Sec24}, while Section~\ref{Sec26} presents a concise comparison of them.

%[Comment Frits]: should we not say a bit on properties of the methods? Or refer to a paper that does so? For instance in the eigenvector method, one could say that, when applied to a complete RR, the outcome is not necessarily the same as the point-based ranking.

\subsection{The direct ranking method} \label{Sec21}

The direct ranking method assumes that the ranking is proportional to both the performance of teams against each of its opponents, as well as to the strength of its opponents. This ranking method (which goes back to \citet{Landau1895}) is obtained by solving the following linear system $Ar = \lambda r$:
\begin{gather} \label{eq1}
\begin{pmatrix}
    \frac{a_{11}}{t_1} & \frac{a_{12}}{t_1} & \dots & \frac{a_{1n}}{t_1}\\
    \frac{a_{21}}{t_2} & \frac{a_{22}}{t_2} & \dots & \frac{a_{2n}}{t_2}\\
    \dots & \dots & \dots & \dots \\
   \frac{a_{n1}}{t_n} & \frac{a_{n2}}{t_n} & \dots & \frac{a_{nn}}{t_n}\\
\end{pmatrix}
\begin{pmatrix}
    r_1 \\
    r_2 \\
    \dots \\
    r_n 
\end{pmatrix}
=
\lambda
\begin{pmatrix}
    r_1 \\
    r_2 \\
    \dots \\
    r_n 
\end{pmatrix}
\end{gather}
where $t_i$ is the number of games played by team $i$ (in our case, $t_i=8$ for all $i$), and $a_{ij}$ is the result of the match between teams $i$ and $j$ for team $i$. In particular, $a_{ij} = 3$ if $i$ wins, $a_{ij} = 1$ if the result is a draw, and $a_{ij} = 0$ if $i$ loses. The vector $(r_1, r_2, \dots, r_n)^T$ represents the strength ratings of the teams, which are used to determine the ranking. Note that this is a right eigenvector of $A$ and $\lambda$ is the associated eigenvalue.

\citet{Landau1895} motivates Equation~\ref{eq1} as follows. Let $r_i$ be the \emph{unknown} quality of player $i$, which is proportional to the expected value of its performance against all the opponents $\sum_{j=1}^n a_{ij}/t_i r_j$; of course, a better result against the same player or the same result against a player of a higher quality is beneficial.

A basic requirement for the direct ranking method to work is that the matrix $A$ should be irreducible, or, equivalently, that the graph which contains a node for each row/column and an arc from node $i$ to column $j$ if and only if $a_{ij}$ is nonzero, is a strongly connected graph. In other words, a path needs to exist from each node to any other node.

However, matrix $A$ is \emph{not} necessarily irreducible; indeed, this is the case in the 2024/25 UEFA Champions League league phase as two teams (\emph{S.~Bratislava} and \emph{Young Boys}) lost all their matches. We apply the following trick to arrive at an irreducible matrix $A$. We add two dummy teams, say teams B and C. Team B beats all other non-dummy teams, and team C loses against all other non-dummy teams, while team C beats team B. It is easy to see that the resulting matrix $A$ is now irreducible; of course, the dummy teams are removed from the ranking.

More information on the direct ranking method can be found in \citet{keener1993perron}, \citet{LambersSpieksma2020}, and \citet{SinnZiegler2022}.

\subsection{Colley} \label{Sec22}

Colley's ranking method aims to rank the teams based on their winning percentage and the strength of the schedule. This ranking is obtained by solving the following linear system of equations:
\begin{gather} \label{eq2}
\begin{pmatrix}
    2 + t_1 & -n_{12} & \dots & -n_{1n}\\
    -n_{21} & 2+t_2 &  \dots & -n_{2n} \\
    \dots & \dots & \dots & \dots \\
    -n_{n1} & -n_{n2} & \dots & 2 + t_{n} 
\end{pmatrix}
\begin{pmatrix}
    r_1 \\
    r_2 \\
    \dots \\
    r_n 
\end{pmatrix}
=
\begin{pmatrix}
    b_1 \\
    b_2 \\
    \dots \\
    b_n 
\end{pmatrix},
\end{gather}
where $n_{ij} = n_{ji}$ is the number of matches between teams $i$ and $j$ (in our case, this is either 0 or 1), and $b_i = 1 + \left( w_i-\ell_i \right) / 2$ with $w_i$ being the number of wins and $\ell_i$ being the number of losses. More information can be found in \citet{colley2002colley}, \citet{chartier2011sensitivity}, and \citet{vaziri2018properties}.

%where $b_i = 1 + \left( w_i-\ell_i \right) / 2$ with $w_i$ being the number of wins plus half the number of draws, and $\ell_i$ being the number of losses plus half the number of draws.
%We don't have to include the 'number of draws' as they are added to both \(w_i, \ell_i\) and \(b_i\) depends on the difference between these two.

\subsection{Generalized Row Sum} \label{Sec23}

The generalized row sum (GRS) is a parametric family of ranking methods that adjusts the average number of points per game by taking into account the performance of the opponents. Specifically, a ranking is obtained by solving the following linear system:
\begin{gather} \label{eq3}
\left[ 
\begin{pmatrix}
    1 & 0 & \dots & 0 \\
    0 & 1 & \dots & 0 \\
    \dots & \dots & \dots & \dots \\
    0 & 0 & \dots & 1
\end{pmatrix}
+
\varepsilon
\begin{pmatrix}
    t_1 & -n_{12} & \dots & -n_{1n}\\
    -n_{21} & t_2 & \dots & -n_{2n} \\
    \dots & \dots & \dots & \dots \\
    -n_{n1} & -n_{n2} & \dots & t_{n} 
\end{pmatrix}
\right]
\begin{pmatrix}
    r_1 \\
    r_2 \\
    \dots \\
    r_n 
\end{pmatrix}
=
\begin{pmatrix}
    s_1 \\
    s_2 \\
    \dots \\
    s_n 
\end{pmatrix}
\end{gather}
where $s_i$ is the \emph{normalized} number of points of team $i$, $\varepsilon \geq 0$ is a parameter.

The points are normalized by subtracting the average number of points of all teams to get a score vector whose elements sum up to zero. This is required because GRS is originally defined in a setting where a win is worth $+1$ point, a draw is worth $0$ points, and a loss is worth $-1$ point \citep{Chebotarev1994, gonzalez2014paired}. However, the UEFA Champions League uses the $\{ 3,1,0 \}$ scoring system, the current standard in association football. Subtracting the average guarantees that $\varepsilon = 0$ results in the ranking given by the number of points collected in all matches \citep{Csato2021d}. In the UEFA Champions League league phase, the average number of points equals $(144 \times 3 - d)/36 = 12 - d/36$, where $d$ is the number of draws in the season. For example, in the 2024/25 season, the normalized score of \emph{Liverpool} is $21 - (12 - 18/36) = 9.5$ since $d = 18$.

On the left-hand side of the system of equations~\eqref{eq3}, the first term corresponds to the number of points obtained by the teams, while the second term reflects their strength of schedule. In particular, parameter $\varepsilon$ is multiplied by the Laplacian matrix of the (undirected) graph where the edges represent the matches played. Therefore, the rating $r_i$ of team $i$ becomes larger (smaller) than $s_i$ if its strength of schedule is above (below) the average strength of schedule. Naturally, the strength of schedule accounts for not only the strength of the direct opponents, but also the strength of the opponents of opponents, and so on.

If $\varepsilon = 0$, GRS is equivalent to the point-based ranking without any tie-breaking criterion. Strength of schedule can be used as a tie-breaking rule if the value of parameter $\varepsilon$ is small. However, higher values of parameter $\varepsilon$ may swap teams that have different numbers of points. More details about this method can be found in \citet{Chebotarev1994}, \citet{gonzalez2014paired}, \citet{Csato2017a}, and \citet{Csato2021d}.

A crucial property of the generalized row sum method is that it implies the same ranking for any value of $\varepsilon$ as the number of points if each team plays the same number of games against any other team, namely, if the tournament is (complete) round robin. \citet{LeivaBertran2025} has recently proposed a new family of generalized win percentage scoring methods for ranking teams in incomplete tournaments, which satisfies a number of reasonable theoretical properties (e.g.\ win dominance, win/loss fairness). This method coincides with the generalized row sum if all teams play the same number of matches \citep[Proposition~1]{LeivaBertran2025} as in the UEFA Champions League league phase.

\subsection{Least squares} \label{Sec24}

Another commonly used ranking method to accommodate incomplete schedules is the least squares method \citep{Stefani1977, Stefani1980, LasekSzlavikBhulai2013}. This involves minimizing the sum of squared errors such that the average rating is equal to zero.
%Let $T$ denote the set of teams and define $N = |T|$. Moreover, let $S$ be the set of all matches. Define the variable $r_i$ which is the rating of team $i \in N$. Moreover, 
Let $S$ be the set of all matches that were played. Let $m_{ij}$ be the goal difference between the home team $i$ and the away team $j$ for each match $(i,j) \in S$. We then solve the following optimization problem:
\begin{equation} \label{eq4}
\begin{alignedat}{2}
    \text{min} \quad 
    & \sum_{(i,j) \in S} \left[ (r_{i} - r_{j}) - m_{ij} \right]^2 \\
    \text{s.t.} \quad 
    &\sum_{i=1}^{N} r_i = 0  && \\
    & r_i \in \mathbb{R}. && 
\end{alignedat}  
\end{equation}
The rating $r_i$ can be interpreted as the expected goal difference of team $i$ against an average team in the competition. \citet{LasekSzlavikBhulai2013} show that the least squares method has good predictive performance. For more details, we refer to \citet{LaprePalazzolo2022}.

\subsection{Some connections between ranking methods} \label{Sec25}

Here we show a relationship between Colley and GRS (Section~\ref{Sec251}), as well as between GRS and least squares (Section~\ref{Sec252}) methods.

\subsubsection{The Colley and Generalized Row Sum methods} \label{Sec251}

\citet{LeivaBertran2025} introduces the family of generalized win percentage scoring methods, which is equivalent to the GRS if all teams play the same number of matches, and is ``similar to the Colley matrix method'' for a particular value of its parameter. This section uncovers what this similarity exactly means.

We refer to a setting where a win gives 2 points, a draw gives 1 point, and a loss gives 0 points, as the $\{ 2,1,0 \}$ scoring system. For such a scoring system, the following statement can be derived.

\begin{proposition}
Under the $\{ 2,1,0 \}$ scoring system, if every team plays the same number of matches $t$, the GRS ranking with $\varepsilon = 1/2$ is identical to the Colley ranking.  
\end{proposition}

\begin{proof}
Let \begin{gather}
L = 
\begin{pmatrix}
    t & -n_{12} & \dots & -n_{1n}\\
    -n_{21} & t & \dots & -n_{2n} \\
    \dots & \dots & \dots & \dots \\
    -n_{n1} & -n_{n2} & \dots & t 
\end{pmatrix}.
\end{gather}
Observe that, for $\varepsilon = 1/2$, the system of equations~\eqref{eq3} that defines GRS can be written as $(I + L/2) r = s$, where the vector $r$ denotes the variables.
Furthermore, the system of equations~\eqref{eq2} that defines the Colley ranking can be written as $(2I + L) q = b$, which is equivalent to $2(I + L/2) q = b$, where the vector $q$ denotes the variables.

The solutions $r^*$ and $q^*$ to \eqref{eq3} and \eqref{eq2}, respectively, are proved to be related: $q_i^* = r_i^*/4 - (t+c-2)/4$ for each $1 \leq i \leq n$, where $c = -\sum_{i=1}^n \left( 2w_i + d_i \right)/t$. More compactly stated:
\begin{equation} \label{eq5}
q^* = r^*/4 - (t+c-2)/4.
\end{equation}
Since the term $(t+c-2)/4$ is constant, the rankings corresponding to vectors $q^*$ and $r^*$ coincide. The rest of the proof consists of verifying this claim.

First, we establish a relation between the right-hand sides of \eqref{eq3} and \eqref{eq2}, i.e., between vectors $s$ and $b$. The definition of $s$ (see Section~\ref{Sec23}) together with the $\{ 2,1,0 \}$ scoring system imply that $s_i=2w_i+d_i + c$ for each team $i$, where $c = -\sum_{i=1}^n \left( 2w_i + d_i \right)/t$. In addition, $t = w_i+d_i + \ell_i$ for each team $i$.
Subtracting this equality from both sides gives $s_i-t = w_i-\ell_i +c$ for all $1 \leq i \leq n$. Since $b_i = 1+ (w_i -\ell_i)/2$ for each $i$ (see Section~\ref{Sec22}), $2b_i = w_i - \ell_i + 2 = s_i - t - c + 2$ holds for each $i$, or in vector notation:
\begin{equation} \label{eq6}
2b = s - t - c + 2.
\end{equation}

Let $q^*$ be a solution to~\eqref{eq2}, i.e., $2(I+L/2)q^* = b$, which is equivalent to 
\begin{equation} \label{eq7}
(I + L/2)4q^*=2b.
\end{equation}
Substituting~\eqref{eq5} and~\eqref{eq6} into~\eqref{eq7} gives: 
\begin{equation} \label{eq8}
(I+L/2)(r^* -t -c+2) = s-t-c+2.
\end{equation}
The row sums of matrix $L$ equal 0, and hence $Le = 0$ for any constant vector $e$. Thus, the left hand side of~\eqref{eq8} can be written as:
\begin{equation*}
(I+L/2)(r^* - t - c + 2) = (I+L/2) r^* + (I+L/2)(-t-c+2) = (I+L/2) r^* + (-t-c+2).
\end{equation*}
This expression equals $s-t-c+2$ according to \eqref{eq8}, thus, $(I+L/2)r^* = s$, and the claim is proven.
\end{proof}

\subsubsection{The Generalized Row Sum and least squares methods} \label{Sec252}

The GRS ranking can also be derived as the optimal solution of a least squares optimization problem if $\varepsilon \to \infty$. Let $h_{ij} = p_{ij} - \bar{p}$ for all matches $(i,j) \in S$ that were played, where $p_{ij}$ is the number of points scored by team $i$ against team $j$ in their match, and $\bar{p}$ is the average number of points per team in a match. Hence, $\bar{p} = (3 \times 144 - d)/(2 \times 144) = 3/2 - d/288$ in the UEFA Champions League league phase. In particular, $d = 18$ implies $\bar{p} = 23/26 = 1.4375$ in the 2024/25 season.

The GRS rating vector $\left[ r_1, r_2, \dots r_n \right]^T$ with $\varepsilon \to \infty$ and $\sum_{i=1}^n r_i = 0$ is the unique minimizer of the following optimization problem \citep[p.~144]{gonzalez2014paired}:
\begin{alignat*}{2}
    \text{min} \quad 
    & \sum_{(i,j) \in S} \left[ (x_{i} - x_{j}) - h_{ij} \right]^2 \\
    \text{s.t.} \quad 
    &\sum_{i=1}^{N} x_i = 0  && \\
    & x_i \in \mathbb{R}. && 
\end{alignat*}
Thus, the only difference between GRS with $\varepsilon \to \infty$ and least squares (the system of equations~\eqref{eq4}) is that the former assumes that the rating is the expected normalized number of points scored against an average team, while the latter assumes that the rating is the expected goal difference against an average team.

\subsection{Comparison of the ranking methods} \label{Sec26}

When applied to a complete round robin tournament where each team plays the same number of matches against all other teams, the ranking obtained from the direct ranking method is not necessarily the same as the point-based ranking. In particular, the direct ranking method favors teams that perform well against high-ranked teams and poorly against low-ranked teams, compared to teams that perform well against low-ranked teams and poorly against high-ranked teams.

The main advantage of GRS resides in its good axiomatic properties \citep{Chebotarev1994, gonzalez2014paired}: it retains several conditions that are satisfied by the ranking obtained from the number of points in a round robin tournament \citep{Csato2021d}. For instance, the GRS ranking coincides with the ranking based on the number of points if parameter $\varepsilon$ equals zero. Consequently, GRS provides an ideal tie-breaking rule with a small value of $\varepsilon$ based on the strength of schedule, even though this might be somewhat unfair since the teams have no control over the performance of their opponents in contrast to goal difference.

According to our knowledge, there are no studies on how to choose the value of parameter $\varepsilon$ in incomplete round robin tournaments. Therefore, we see no particular reason why to pick up the Colley method, which would be a special case of GRS if UEFA would use the $\{ 2,1,0 \}$ scoring system according to Section~\ref{Sec251}. Furthermore, since UEFA uses the $\{ 3,1,0 \}$ scoring system, the original variant of the Colley method presented in Section~\ref{Sec22} does not imply the same ranking as the official point-based ranking.

GRS can also be used to establish a dominance relation between two teams; it would be difficult to argue for ranking team $i$ above team $j$ if $r_i (\varepsilon) \leq r_j (\varepsilon)$ holds for any $\varepsilon \geq 0$. In contrast to the direct ranking method, the ranking according to GRS does not depend on whether a team performs well against high- or low-ranked teams since the sets of opponents and the number of points scored by them are treated separately in the system of equations~\eqref{eq3}.

In contrast to the official ranking, the GRS is not influenced by goal difference. This could be a disadvantage by removing any incentives to win by more goals, and can call for the least squares method based on goal differences as formulated in Section~\ref{Sec24}. But the latter approach seems to give an excessive role to goal difference at the expense of the trichotomous match outcome.

\section{Results}

\begin{table}[t!]
    \centering
    \caption{Rankings with three alternative methods in the 2024/25 UEFA Champions League league phase}
    \label{Table2}
    \begin{threeparttable}
	\rowcolors{1}{}{gray!20}
       \begin{tabularx}{\textwidth}{l CCC CCCC} \toprule
    Method & Points & Official ranking  & Direct & Colley & GRS ($\varepsilon=0.01$) & GRS ($\varepsilon \to \infty$) & Least squares \\ \bottomrule
    Liverpool & 21    & 1     & 1     & 1 & 1     & 1     & 3 \\
    Barcelona & 19    & 2     & 2     & 3 & 3     & 2     & 1 \\
    Arsenal & 19    & 3     & 3     & 2 & 2     & 3     & 2 \\
    Inter Milan & 19    & 4     & 5     & 4 & 4     & 6     & 8 \\
    Atleti & 18    & 5     & 10    & 8 & 5     & 8     & 18 \\
    Leverkusen & 16    & 6     & 6     & 6 & 6     & 5     & 12 \\
    Lille & 16    & 7     & 4     & 5 & 7     & 4     & 6 \\
    Aston Villa & 16    & 8     & 13    & 11 & 8     & 15    & 11 \\
    Atalanta & 15    & 9     & 18    & 9 & 11    & 14    & 4 \\
    B.~Dortmund & 15    & 10    & 14    & 12 & 12    & 12    & 5 \\
    Real Madrid & 15    & 11    & 7     & 7 & 9     & 7     & 7 \\
    Bayern M\"unchen & 15    & 12    & 11    & 13 & 10    & 10    & 9 \\
    Milan & 15    & 13    & 17    & 16 & 13    & 13    & 20 \\
    PSV   & 14    & 14    & 12    & 15 & 14    & 16    & 14 \\
    Paris & 13    & 15    & 22    & 18 & 17    & 18    & 15 \\
    Benfica & 13    & 16    & 8     & 14 & 16    & 11    & 10 \\
    Monaco & 13    & 17    & 9     & 10 & 15    & 9     & 13 \\
    Brest & 13    & 18    & 21    & 19 & 18    & 17    & 23 \\
    Feyenoord & 13    & 19    & 16    & 20 & 19    & 19    & 27 \\
    Juventus & 12    & 20    & 15    & 17 & 20    & 20    & 16 \\
    Celtic & 12    & 21    & 25    & 24 & 21    & 25    & 25 \\
    Man City & 11    & 22    & 23    & 25 & 25    & 24    & 21 \\
    Sporting CP & 11    & 23    & 19    & 22 & 24    & 23    & 17 \\
    Club Brugge & 11    & 24    & 20    & 21 & 22    & 22    & 22 \\
    GNK Dinamo & 11    & 25    & 24    & 23 & 23    & 21    & 32 \\
    Stuttgart & 10    & 26    & 27    & 28 & 26    & 29    & 31 \\
    Shakhtar & 7     & 27    & 28    & 27 & 27    & 27    & 26 \\
    Bologna & 6     & 28    & 26    & 26 & 28    & 26    & 19 \\
    Crvena Zvezda & 6     & 29    & 30    & 29 & 29    & 28    & 29 \\
    Sturm Graz & 6     & 30    & 32    & 30 & 30    & 30    & 28 \\
    Sparta Praha & 4     & 31    & 33    & 31 & 31    & 32    & 33 \\
    Leipzig & 3     & 32    & 31    & 32 & 32    & 31    & 24 \\
    Girona & 3     & 33    & 34    & 34 & 34    & 34    & 30 \\
    Salzburg & 3     & 34    & 29    & 33 & 33    & 33    & 36 \\
    S.~Bratislava & 0     & 35    & 35    & 36 & 36    & 36    & 35 \\
    Young Boys & 0     & 36    & 36    & 35 & 35    & 35    & 34 \\ \bottomrule
    \end{tabularx}
    \begin{tablenotes} \footnotesize
        \item The second column Ranking shows the position in the official league table.
        \item GRS stands for Generalized row sum.
    \end{tablenotes}
    \end{threeparttable}
\end{table}

% [Comment Frits] This phrasing takes the point-based ranking implicitly as default. It seems more appropriate to discuss first the differences of the 6 columns wrt the 3 categories of teams. In fact, the Celtic and Dynamo case should be discussed first.

The rankings produced by the alternative ranking methods in Section~\ref{Sec2} are shown in Table~\ref{Table2}. For the generalized row sum, we provide the two extremal results: when the strength of schedule is used only for tie-breaking ($\varepsilon$=0.01) and when the performance of the opponents has the highest possible effect ($\varepsilon \to \infty$). The rankings for other values of the parameter $\varepsilon$ can be found in the Appendix, in Table~\ref{TableA1}.

%Similar to the Generalized Row Sum method, the Eigenvector also takes into account the strength of schedule. Therefore, Real Madrid and Benfica are ranked in the Top 8, at the expense of Atletico Madrid and Aston Villa. The result of the Colley ranking method is similar to the official ranking. In contrast, the Massey method produces quite a different ranking since it is primarily based on goal differences. For example, it ranks Leverkusen only in the 19$th$ place instead of the official sixth position. 

Clearly, different methods produce different rankings.
However, while the direct, Colley, and Generalized Row Sum methods tend to agree, the least squares method give some surprising results. For example, \emph{Atleti}, with 18 points, is ranked behind \emph{Sporting CP} that scored only 11 points, and is directly followed by \emph{Bologna}, with only 6 points. This can be explained by its focus on goal differences as discussed in Section~\ref{Sec26}: the opponents of \emph{Atleti} and \emph{Bologna} have a cumulative goal difference of $-33$ and $+42$, respectively. Consequently, least squares indicates that \emph{Atleti} (\emph{Bologna}) had an exceptionally (dis)favourable schedule, which is only partially supported by the other methods based on the number of points scored. Therefore, the ranking produced by the least squares will be treated separately in the following discussion.

With respect to qualification, the clubs are essentially allocated into three sets: ranked in the Top 8 (direct qualification for the Round of 16), ranked between 9th and 24th (qualification for the knockout phase play-offs), and ranked between 25th and 36th (elimination).
In view of these thresholds, the results can be summarized as follows:
\begin{itemize}
\item
\emph{Aston Villa} falls out of the Top 8 when the strength of schedule has a higher impact. In particular, it is only ranked 13th in the direct ranking method and even only 15th in the GRS with $\varepsilon \to \infty$. On the other hand, \emph{Real Madrid} is ranked 7th according to the direct method and GRS ($\varepsilon \to \infty$). Hence, it can be argued that \emph{Real Madrid} should have directly qualified for the Round of 16 at the expense of \emph{Aston Villa}, even though the latter club has scored one point more.
\item
\emph{Monaco} is clearly the best team among the five with 13 points if the schedule of strength is taken into account, and it can even be ranked higher than some teams with 15 points. Similarly, \emph{Benfica} could also be ranked higher; it is placed in the Top 8 by the direct method.
\item
All alternative rankings except GRS if the strength of schedule plays a minor role ($\varepsilon < 0.5$) exclude \emph{Celtic} from the Top 24. However, the Croatian champion \emph{GNK Dinamo} is eliminated only according to the official ranking rule; if its strength of schedule is taken into account, it is never ranked below the 23rd position.
%\item
%\emph{Man City} is ranked in the Top 24 except by the GRS with a small value of parameter $\varepsilon$, which shows that it has had the easiest schedule among the clubs with 11 points.
\end{itemize}
Compared to the official ranking, the least squares method removes three teams from the Top 8 and places three teams with 15 points directly in this set. Furthermore, \emph{Feyenoord} would be eliminated instead of \emph{Bologna}, which would be difficult to explain for the stakeholders since the former team has 13 points rather than 6. Last but not least, even \emph{Leipzig} would qualify for the play-offs despite its sole win against \emph{Sporting CP}.

Nonetheless, some teams are consistently ranked in the same position:
\begin{itemize}
\item
All ranking methods agree on the top three teams, which are \emph{Liverpool}, \emph{Barcelona}, and \emph{Arsenal}.
\item
In the GRS, seven teams have a guaranteed place in the Top 8 independently of the value of parameter $\varepsilon$.
\item
All ranking methods agree on the elimination of nine teams (\emph{Stuttgart, Shakhtar, Crvena Zvezda, Sturm Graz, Sparta Praha, Girona, Salzburg, S.~Bratislava}, and \emph{Young Boys}).
\end{itemize}

In some cases, the ranking methods produce remarkably different results even if the least squares method is disregarded. For example, \emph{Atalanta} is officially ranked 9th, while it is ranked 18th according to the direct method.
%In general, the Massey method produces the most distinctive ranking. For example, it is the only method that places Barcelona over Liverpool, even though Liverpool has two points more than Barcelona. Moreover, Leverkusen is only ranked 19th (6th in the official ranking), Feyenoord is ranked only 32th (19th in the official ranking) and Leipzig is ranked 22nd (but is only ranked 32th in the official ranking). The reason is that Massey's ranking is solely based on goal difference and completely ignores wins.
The result of the Colley ranking method is quite similar to the official ranking, which can be explained by its strong relation to the GRS as presented in Section~\ref{Sec251}.

As for the GRS ranking, playing only eight rounds in a league with 36 teams seems to be insufficient to obtain a unique ranking. In contrast, in the 2019/20 season of the double round robin European football leagues, where 18--20 teams played 25--29 rounds before the COVID-19 pandemic, the ranking with the generalized row sum method was essentially insensitive to the value of parameter $\varepsilon$ \citep{Csato2021d}. Indeed, playing a higher number of rounds increases the variance of the number of points and decreases the variance in the strength of opponents, reducing the probability that a team with fewer points scored against stronger opponents is ranked above a team with more points scored against weaker opponents.

\section{Discussion}

Ranking teams in an incomplete round robin tournament is challenging since it is difficult to choose an uncontroversial balance between the number of points scored and the strength of schedule. Our results reveal that this issue is also relevant for the 2024/25 UEFA Champions League league phase: both the set of teams directly qualified for the Round of 16 and the set of teams eliminated would have changed when using reasonable alternative ranking methods. Consequently, the current mechanism using four pots based on strengths according to the UEFA club coefficients may not sufficiently reduce the impact of different strengths of opponent sets.

Possible remedies include
(1) using a ranking that reflects the strengths of the teams better than UEFA club coefficients to assign the teams to pots (e.g.\ a ranking based on Elo ratings as demonstrated by \citet{Csato2024c}), and
(2) using the strength of schedule as the primary tie-breaking criterion instead of goal difference (which would boil down to the GRS ranking with, for example, $\varepsilon=0.01$). 
Both recommendations seem to be realistic since the Elo approach has been adopted for the FIFA World Ranking in 2018 \citep{FIFA2018c} and the Buchholz method, 
based on the arithmetic mean of the points scored by opponents, is the most common tie-breaking rule in Swiss-system tournaments \citep{Freixas2022}.

% A more drastic solution could be to use 9 strength-based pots, each of size 4, with each team playing one team from each pot. The latter would result in 9 matches per team, but would reduce the variability in strength of schedule.

\section*{Acknowledgements}
\addcontentsline{toc}{section}{Acknowledgements}
\noindent
Four reviewers (\emph{Julien Guyon}, \emph{Mehmet S.~Ismail}, \emph{Michael A.~Lapr\'e}, \emph{Fernando Leiva Bertr\'an}), and \emph{Kees Pippel} provided valuable remarks on an earlier draft. \\
The research was supported by the National Research, Development and Innovation Office under Grants FK 145838, by the J\'anos Bolyai Research Scholarship of the Hungarian Academy of Sciences, and by the Research Foundation - Flanders (FWO) under grant no.~S005521N.

\bibliographystyle{apalike} 
\bibliography{bibliography}

\clearpage

\section*{Appendix}
\renewcommand{\thetable}{A.\arabic{table}}
\setcounter{table}{0}

\begin{table}[ht!]
  \centering
   \caption{Rankings with the Generalized Row Sum method in the 2024/25 UEFA Champions League league phase}
   \label{TableA1}
\centerline{
\begin{threeparttable}
	\rowcolors{1}{}{gray!20}	
    \begin{tabularx}{1\textwidth}{lc CCCC CCCC CCCC} \toprule
    Parameter ($\varepsilon$) & Ranking & 0     & 0.01  & 0.1   & 0.25  & 0.5   & 0.75  & 1     & 2     & 5     & 10    & 100   & $\infty$ \\ \bottomrule
    Liverpool & 1     & 1     & 1     & 1     & 1     & 1     & 1     & 1     & 1     & 1     & 1     & 1     & 1 \\
    Barcelona & 2     & 2     & 3     & 3     & 2     & 2     & 2     & 2     & 2     & 2     & 2     & 2     & 2 \\
    Arsenal & 3     & 2     & 2     & 2     & 3     & 3     & 3     & 3     & 3     & 3     & 3     & 3     & 3 \\
    Inter Milan & 4     & 2     & 4     & 4     & 4     & 4     & 4     & 4     & 5     & 6     & 6     & 6     & 6 \\
    Atleti & 5     & 5     & 5     & 5     & 7     & 8     & 8     & 8     & 8     & 8     & 8     & 8     & 8 \\
    Leverkusen & 6     & 6     & 6     & 6     & 5     & 5     & 6     & 6     & 6     & 5     & 5     & 5     & 5 \\
    Lille & 7     & 6     & 7     & 7     & 6     & 6     & 5     & 5     & 4     & 4     & 4     & 4     & 4 \\
    Aston Villa & 8     & 6     & 8     & 9     & 11    & 13    & 15    & 15    & 15    & 15    & 15    & 15    & 15 \\
    Atalanta & 9     & 9     & 11    & 12    & 13    & 14    & 14    & 14    & 14    & 14    & 14    & 14    & 14 \\
    B.~Dortmund & 10    & 9     & 12    & 11    & 10    & 11    & 11    & 11    & 11    & 12    & 12    & 12    & 12 \\
    Real Madrid & 11    & 9     & 9     & 8     & 8     & 7     & 7     & 7     & 7     & 7     & 7     & 7     & 7 \\
    Bayern M\"unchen & 12    & 9     & 10    & 10    & 9     & 9     & 10    & 10    & 10    & 10    & 10    & 10    & 10 \\
    Milan & 13    & 9     & 13    & 13    & 12    & 12    & 12    & 12    & 13    & 13    & 13    & 13    & 13 \\
    PSV   & 14    & 14    & 14    & 15    & 16    & 16    & 16    & 16    & 16    & 16    & 16    & 16    & 16 \\
    Paris & 15    & 15    & 17    & 17    & 17    & 17    & 17    & 18    & 18    & 18    & 18    & 18    & 18 \\
    Benfica & 16    & 15    & 16    & 16    & 15    & 15    & 13    & 13    & 12    & 11    & 11    & 11    & 11 \\
    Monaco & 17    & 15    & 15    & 14    & 14    & 10    & 9     & 9     & 9     & 9     & 9     & 9     & 9 \\
    Brest & 18    & 15    & 18    & 18    & 18    & 18    & 18    & 17    & 17    & 17    & 17    & 17    & 17 \\
    Feyenoord & 19    & 15    & 19    & 19    & 19    & 19    & 19    & 19    & 19    & 19    & 19    & 19    & 19 \\
    Juventus & 20    & 20    & 20    & 20    & 20    & 20    & 20    & 20    & 20    & 20    & 20    & 20    & 20 \\
    Celtic & 21    & 20    & 21    & 23    & 24    & 25    & 25    & 25    & 25    & 25    & 25    & 25    & 25 \\
    Man City & 22    & 22    & 25    & 25    & 25    & 24    & 24    & 24    & 24    & 24    & 24    & 24    & 24 \\
    Sporting CP & 23    & 22    & 24    & 24    & 23    & 23    & 23    & 23    & 23    & 23    & 23    & 23    & 23 \\
    Club Brugge & 24    & 22    & 22    & 21    & 21    & 21    & 22    & 22    & 22    & 22    & 22    & 22    & 22 \\
    GNK Dinamo & 25    & 22    & 23    & 22    & 22    & 22    & 21    & 21    & 21    & 21    & 21    & 21    & 21 \\
    Stuttgart & 26    & 26    & 26    & 26    & 26    & 28    & 28    & 28    & 29    & 29    & 29    & 29    & 29 \\
    Shakhtar & 27    & 27    & 27    & 27    & 28    & 27    & 27    & 27    & 27    & 27    & 27    & 27    & 27 \\
    Bologna & 28    & 28    & 28    & 28    & 27    & 26    & 26    & 26    & 26    & 26    & 26    & 26    & 26 \\
    Crvena Zvezda & 29    & 28    & 29    & 29    & 29    & 29    & 29    & 29    & 28    & 28    & 28    & 28    & 28 \\
    Sturm Graz & 30    & 28    & 30    & 30    & 30    & 30    & 30    & 30    & 30    & 30    & 30    & 30    & 30 \\
    Sparta Praha & 31    & 31    & 31    & 31    & 31    & 32    & 32    & 32    & 32    & 32    & 32    & 32    & 32 \\
    Leipzig & 32    & 32    & 32    & 32    & 32    & 31    & 31    & 31    & 31    & 31    & 31    & 31    & 31 \\
    Girona & 33    & 32    & 34    & 34    & 34    & 34    & 34    & 34    & 34    & 34    & 34    & 34    & 34 \\
    Salzburg & 34    & 32    & 33    & 33    & 33    & 33    & 33    & 33    & 33    & 33    & 33    & 33    & 33 \\
    S.~Bratislava & 35    & 35    & 36    & 36    & 36    & 36    & 36    & 36    & 36    & 36    & 36    & 36    & 36 \\
    Young Boys & 36    & 35    & 35    & 35    & 35    & 35    & 35    & 35    & 35    & 35    & 35    & 35    & 35 \\ \bottomrule
    \end{tabularx}
    \begin{tablenotes} \footnotesize
    	%\item Probabilities are based on 1 million simulation runs.
    	\item Teams are ranked according to the official league phase table.
    	\item The second column Ranking shows the position in the official league table.
    \end{tablenotes}
\end{threeparttable}
}
\end{table}

\end{document}